\documentclass{ejm2}
\usepackage{latexsym,amsfonts,amsmath,cyracc,verbatim}
\usepackage{graphicx}

\title[Corrector estimates of a locally-periodic scenario]{Corrector estimates for the homogenization of a locally-periodic medium with areas of low and high diffusivity}

\author[Adrian Muntean and Tycho van Noorden]{
 A.\ns M\ls U\ls N\ls T\ls E\ls A\ls N$\,^{1,2}$\ns \and T.\ns L.\ns V\ls A\ls N\ns N\ls O\ls O\ls R\ls D\ls E\ls N$\,^3$}

\affiliation{$^1\,$CASA - Centre for Analysis, Scientific computing and Applications,\\
 Department of Mathematics and Computer Science,\\
Technische Universiteit Eindhoven, P.O. Box 513, 5600 MB Eindhoven,
The Netherlands\\
$^2\,$Institute for Complex Molecular Systems (ICMS), Technische
Universiteit Eindhoven, P.O. Box 513, 5600 MB Eindhoven, The
Netherlands\\
$^3$ Chair of Applied Mathematics 1,
Department of Mathematics, University of Erlangen-N\"urnberg,
Martensstra\ss e 3, 91058 Erlangen, Germany
}

\begin{document}
\maketitle

\begin{abstract}
We prove an upper bound for the convergence rate of the
homogenization limit $\epsilon\to 0$ for a linear transmission
problem for a  advection-diffusion(-reaction) system posed in areas
with low and high diffusivity, where $\epsilon$ is a suitable scale
parameter. On this way, we justify the formal homogenization
asymptotics obtained by us earlier by proving an upper bound for the
convergence rate (a corrector estimate). The main ingredients of the
proof of the corrector estimate include integral estimates for rapidly oscillating functions with prescribed average, properties of the
macroscopic reconstruction operators, energy bounds and extra
two-scale regularity estimates. The whole procedure essentially
relies on a good understanding of the analysis of the limit
two-scale problem.
\end{abstract}



\newtheorem{theorem}{Theorem}[section]
\newdefinition{remark}[theorem]{Remark}
\newdefinition{assumption}{Assumption}
\newdefinition{lemma}[theorem]{Lemma}
\newdefinition{claim}[theorem]{Claim}
\newdefinition{definition}[theorem]{Definition}
\newdefinition{proposition}[theorem]{Proposition}

\section{Introduction}

We study the averaging of a system of reaction-diffusion equations
with linear transmission condition posed in a class of highly
heterogeneous media including areas of low and high diffusivity. Our
aim is twofold: on one hand, we wish to justify rigorously the
formal asymptotics expansions performed in \cite{Tycho_Adrian},
while on the other hand we wish to understand the error caused by
replacing a heterogeneous solution by an approximation (averaged)
estimate. To this end, we prove an upper bound for the convergence
rate of the limit procedure $\epsilon\to 0$, where $\epsilon$ is a
suitable scale parameter (see section \ref{microstructure} for the
definition of $\epsilon$). The materials science scenario we have in
view is motivated by a very practical problem: the sulfate corrosion
of concrete. We refer the reader to
 \cite{Beddoe} for a nice and detailed description of the physico-chemical scenario and to  \cite{Tasnim1} for the formal (two-scale) averaging of
 locally-periodic distributions of unsaturated pores attacked by sulphuric
 acid. The reference
 \cite{Tasnim2} contains the rigorous proof of the limiting procedure $\epsilon\to0$ treating the uniformly periodic case of
 the same corrosion scenario. In \cite{Tasnim2}, the main working tools involve the two-scale convergence concept in the sense of Nguetseng and
 Allaire\footnote{For an introduction to two-scale convergence, see
 \cite{Lukkassen}.}
 combined with  a periodic unfolding of the oscillatory boundary.
 Here, we use an energy-type method.

The framework that we tackle in this context is essentially a
deterministic one. We assume that a distribution of the {\em
locally-periodic array of microstructures}\footnote{Here, we deviate
from the purely periodic setting. It is worth noting also that the
assumption of statistically homogeneity of distribution of
microstructures, which is a crucial restriction of stochastic
homogenization, does not cover
 all possible configurations of locally-periodic geometries. Also, as we formulate the working framework, we cannot treat neither randomly placed microstructures
 nor stochastic distributions of microgeometries. Consequently, the exact connection between these two averaging techniques is hard to make precise.} is known {\em a priori}. We refer the reader to \cite{Capriz} for related discussion of continuum mechanical descriptions of balance laws in continua with microstructure.  At the mathematical level, we succeed to
combine successfully  the philosophy of getting correctors as
explained in the analysis by Chechkin and Piatnitski \cite{Chechkin}
with the intimate two-scale structure of our system; see also
\cite{CAM,CF} for related settings where similar averaging
strategies are used. For methodological hints on how to get corrector estimates, we
refer the reader to the analysis shown in
\cite{Eck_correctors} for the case of a reaction-diffusion phase
field-like system posed in fixed domains. For a higher level of discussion, see the monograph
\cite{CPS} which is a nicely written introductory text on homogenization
methods and applications. It is worth noting that the periodic
homogenization of linear transmission problems is a well-understood
subject (cf. \cite{Conca0,Ene,AG,fast}, e.g.), however much less is
known if one steps away from the periodic setting even if one stays
within the deterministic case. If the microstructures are
distributed in a suitable random fashion, then concepts like random
fields (see \cite{Belyaev_Ef} or \cite{Mikelic_random_fields} and
references cited therein; \cite{Heida2009} and intimate connection
with the stochastic geometry of the perforations) can turn to be
helpful to getting averaged equations (eventually also capturing new
memory terms), but corrector estimates seem to be very hard to get;
cf., e.g., \cite{Bal}. Disordered media for which the stationarity
and/or ergodicity assumptions on the random measure do not hold or
situations where at a given time scale length scales are
non-separated are typical situations that cannot be averaged with
existing techniques (so it makes no sense to search here for correctors).

This paper is organized as follows:
In Section \ref{problem} we introduce the necessary notation, the
microscale and the two-scale limit model.
Section \ref{main_result} contains the main result of our paper --
Theorem \ref{MR}. In Section \ref{technical} we introduce the
technical assumptions needed to obtain the correctors. At this
point, we also collect the known well-posedness and regularity
results for both the microscopic model and the two-scale limit
model. The rest of the paper consists of the proof of the
convergence rate and is the subject of Section \ref{Proofs}.

\section{Statement of the problem}\label{problem}
In this section we introduce our notation, the microscale model and the two-scale limit problem. We start off with the definition of the
locally periodic heterogeneous medium.

\subsection{Definition of the locally-periodic micro medium}\label{microstructure}



\begin{figure}[hbpt]
\begin{center}
\includegraphics[width=8cm]{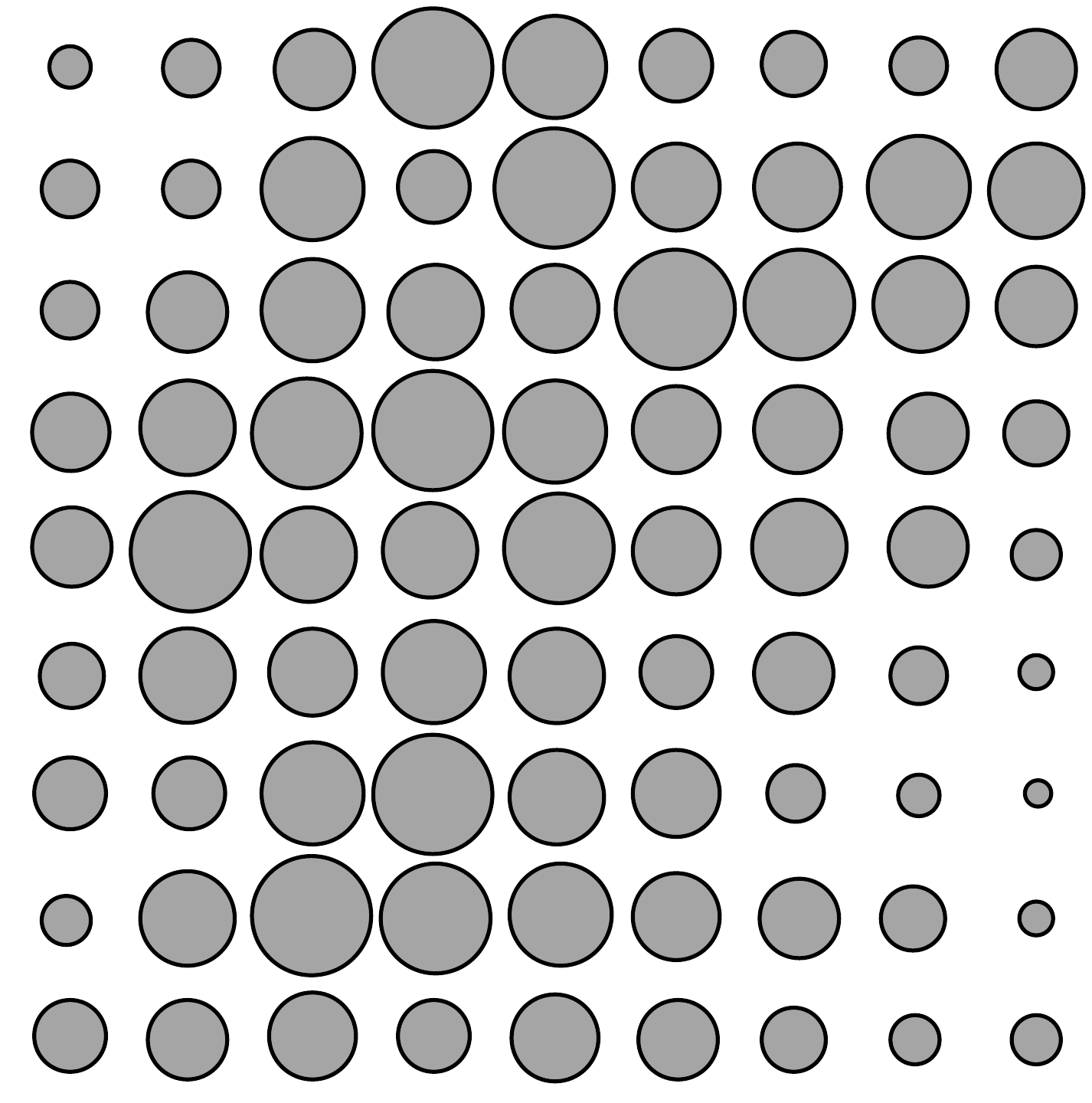}
\end{center}
\caption{Schematic representation of a locally-periodic heterogeneous medium.
The centers of the gray circles are on a grid with width $\epsilon$.
These circles represent the areas of low diffusivity and their radii may vary.
\label{fig1}}
\end{figure}

We consider a heterogenous medium consisting of areas of
high and low diffusivity.
The medium is in the present paper represented by a two dimensional domain. For the definition of the locally-periodic medium with inclusions, we follow in main lines \cite{Chechkin}.
We denote the two dimensional bounded domain by $\Omega\subset \mathbb{R}^2$, with boundary $\Gamma$.
Denote
\begin{align*}
&J^\epsilon:=\{j\in\mathbb{Z}^2\,|\,\mbox{dist}(\epsilon j,\Gamma)\geq\epsilon\sqrt{2}\},\\
&{\cal U}:=\{y\in \mathbb{R}^2\,|\, -1/2\leq y_i \leq 1/2\,\, \mbox{for}\,\, i=1,2\}).
\end{align*}
A convenient way to parameterize the interface $\Gamma^\epsilon$ between the high and low diffusivity areas, is to use a level set function, which we denote by
$S^\epsilon(x)$:
\begin{align*}
x\in \Gamma^\epsilon \Leftrightarrow S^\epsilon(x)=0,
\end{align*}
Since we allow the size and shape of the perforations to vary with the
macroscopic variable $x$, we use the following characterization
of $S^\epsilon$:
\begin{align}
S^\epsilon(x):=S(x,x/\epsilon),
\label{locper}
\end{align}
where $S:\Omega\times {\cal U} \rightarrow \mathbb{R}$ is 1-periodic in its second variable, and where $S$ is independent of $\epsilon$.
We assume that $S(x,0)<const.<0$ and $S(x,y)|_{y\in\partial U}>const.>0$ for all $x\in\Omega$, so that the areas of low diffusivity in each unit cell do not touch each other.
We set
\begin{align*}
Q_j^\epsilon:=\{x\in\epsilon({\cal U}+j)\,|\,S(x,x/\epsilon)<0\},
\end{align*}
and introduce the area of low diffusivity $\Omega_l^\epsilon$ as follows:
\begin{align*}
\Omega_l^\epsilon:=\bigcup_{j\in J^\epsilon}Q_j^\epsilon,
\end{align*}
and the area of high diffusivity $\Omega_h^\epsilon$ as:
\begin{align*}
\Omega_h^\epsilon:=\Omega\backslash\Omega_l^\epsilon.
\end{align*}

We call a medium of which the geometry is specified with a level set function of the type that is given in \eqref{locper} a locally periodic medium \cite{Chechkin}. In Fig.\ \ref{fig1} a schematic
picture is given of how such a medium might look like for a given $\epsilon>0$. Note that by construction the area of high diffusivity $\Omega_h^\epsilon$ is connected and the area of low diffusivity $\Omega_l^\epsilon$ is disconnected.
The interface between high and low diffusivity areas $\Gamma^\epsilon$ is now given by $\Gamma^\epsilon=\partial \Omega_l^\epsilon$ and the boundary of $\Omega_h^\epsilon$ is given by $\partial \Omega_h^\epsilon=\Gamma^\epsilon\cup \Gamma$.

Furthermore, we use the notation $\Omega_1^\epsilon:=\Omega\backslash\bigcup_{j\in J^\epsilon}(\epsilon({\cal U}+j))$, and we introduce for later use the smooth cut-off function $\chi_\epsilon(x)$ that satisfies
$0\leq \chi_\epsilon(x)\leq 1$, $\chi_\epsilon(x)=0$ if $x\in\Omega_1^\epsilon$ and $\chi_\epsilon(x)=1$ if $\mbox{dist}(x,\Omega_1^\epsilon)\geq\mbox{dist}(\Gamma^\epsilon,\Omega_1^\epsilon)$.
Moreover, $\epsilon|\nabla \chi_\epsilon|\leq C$ and $\epsilon^2|\Delta \chi_\epsilon|\leq C$, with $C$ independent of $\epsilon$,
and also
\begin{align}
\|1-\chi_\epsilon\|_{L^2(\Omega)}\leq \epsilon^{1/2} C,\nonumber\\
\|\nabla \chi_\epsilon\|_{L^2(\Omega)}\leq \epsilon^{-1/2} C,\label{boundschi}\\
\|\Delta \chi_\epsilon\|_{L^2(\Omega)}\leq \epsilon^{-3/2} C, \nonumber
\end{align}
with $C$ again independent of $\epsilon$ (see e.g.\ \cite{Chechkin,Eck_correctors}).


In addition, we need to expand the normal $\nu^\epsilon$ to $\Gamma^\epsilon$ in a power series in $\epsilon$.
This can be done in terms of the level set function $S^\epsilon$, which we assume to be sufficiently regular so that all the following computations make sense (see assumption (B1) in section \ref{technical}):
\begin{align}
\nu^\epsilon=\frac{\nabla S^\epsilon(x)}{|\nabla S^{\epsilon}(x)|}=\frac{\nabla S(x,x/\epsilon)}{|\nabla S(x,x/\epsilon)|}=\frac{\nabla_x S+\frac{1}{\epsilon}\nabla_yS}{|\nabla_x S+\frac{1}{\epsilon}\nabla_yS|}
\,\,\, \mbox{at}\,\,\, x\in \Gamma^\epsilon.
\end{align}
First we expand $|\nabla S^{\epsilon}|$. Using
the Taylor series of the
square-root function, we obtain
\begin{eqnarray}
|\nabla S^{\epsilon}|&=&
\frac{1}{\epsilon}|\nabla_{y} S|+O(\epsilon^0). \label{nabsex}
\end{eqnarray}
In the same fashion, we get
\begin{eqnarray*}
\nu^{\epsilon}=\nu_{0}+\epsilon \nu_{1}+O(\epsilon^{2}),
\end{eqnarray*}
where
\begin{eqnarray*}
\nu_{0}:=\frac{\nabla_{y}S}{|\nabla_{y}S|}
\end{eqnarray*}
and
\begin{eqnarray}
\nu_{1}&:=&\frac{\nabla_{x}S}{|\nabla_{y}S|}-
\frac{(\nabla_{x}S\cdot \nabla_{y}S)}{|\nabla_{y}S|^{2}}\frac{\nabla_{y}S}{|\nabla_{y}S|}. \label{nu11}
\end{eqnarray}
If we introduce the normalized tangential vector $\tau_{0}$, with
$\tau_0\perp\nu_{0}$,
we can rewrite $\nu_1$ as
\begin{eqnarray}
\nu_{1}&=&\tau_{0}\frac{\tau_{0}\cdot(\nabla_{x}S)}{|\nabla_{y}S|}.\label{nu1}
\end{eqnarray}
Later on we will also use the notation
\begin{align*}
\nu_1^\epsilon=\epsilon^{-1}(\nu^\epsilon-\nu_0)=\nu_1+O(\epsilon).
\end{align*}

\subsection{Microscopic equations and their upscaled form}

We focus on the following microscopic model
\begin{eqnarray}
&&\begin{cases}
\partial_t u_{\epsilon}=\nabla\cdot(D_h \nabla u_{\epsilon}-q_{\epsilon}u_{\epsilon}) &\\
\end{cases}
\,\,\,\,\,\,\mbox{   in   } \Omega_h^{\epsilon},\label{dimlesseq1}
\\
&&
\begin{cases}
\partial_t v_{\epsilon}=\epsilon^2\nabla\cdot(D_l \nabla v_{\epsilon}) &
\,\,\,\,\,\,\mbox{   in   } \Omega_l^{\epsilon},
\end{cases}\label{dimlesseq1b}
\\
&&\begin{cases}\label{transmission}
\nu^\epsilon\cdot(D_h \nabla u_{\epsilon})=
\epsilon^2\nu^\epsilon\cdot(D_l \nabla v_{\epsilon}) & \\
u_\epsilon=v_\epsilon & \\
\end{cases}
\,\,\,\,\,\, \mbox{   on   } \Gamma^{\epsilon}, \label{dimlesseq2}\\
&&\begin{cases}
u_\epsilon(x,t)=u_b(x,t) & \\
\end{cases} \,\,\,\,\,\,\, \mbox{on}\,\,\, \Gamma,\label{bcg}\\
&&\begin{cases}
u_{\epsilon}(x,0)=u_\epsilon^{I}(x)&\mbox{   in   } \Omega_h^{\epsilon},\\
v_{\epsilon}(x,0)=v_\epsilon^{I}(x)&\mbox{   in   } \Omega_l^{\epsilon}.
\end{cases}
\label{dimlesseq3}
\end{eqnarray}
Here we denoted the tracer concentration in the high diffusivity area by
$u_\epsilon$, the concentration in the low diffusivity area by $v_\epsilon$, and the
velocity of the fluid phase by $q_\epsilon$. Furthermore,
$D_h$ denotes the diffusion coefficient in the high diffusivity
region, $D_l$ the diffusion coefficient in the low diffusivity regions, $\nu^\epsilon$ denotes the unit normal to the boundary $\Gamma^\epsilon(t)$, where  $u_b$ denotes the Dirichlet boundary data for
the concentration $u^\epsilon$, and where
$u^I_\epsilon$ and $v^I_\epsilon$ are the initial values for the active concentrations
$u_\epsilon$ and $v_\epsilon$, respectively.

In order to formulate the upscaled equations and obtain closed-formulae for the effective transport coefficients, we use the notations
\begin{equation}
B(x):=\{y\in {\cal U}: S(x,y)<0\},
\end{equation}
and
\begin{equation}
Y(x):={\cal U}-B(x),
\end{equation}
and we define the following $x$-dependent cell problem:\footnote{See, e.g.\ \cite{CP,hornung}, for the role of cell problems in formulating averaged equations.}
\begin{eqnarray}
\begin{cases}
\Delta_{y}M_{j}(x,y)=0 & \mbox{for all}\,\,\, x\in\Omega,\, y\in Y(x),\\
\nu_{0}\cdot\nabla_{y}M_{j}(x,y)=-\nu_0\cdot e_{j}& \mbox{for all }\ x\in\Omega,\, y\in \partial B(x),\\
\mbox{$M_j(x,y)$ $y$-periodic},&
\end{cases} \label{cellpu}
\end{eqnarray}
for $j\in\{1,2\}$. To ensure the uniqueness of weak solutions to this cell problem, suitable conditions on the spatial averages of the cell functions need to be added.

The solution to this cell problems allows us to write the results of the
formal homogenization procedure in the form of the following
distributed-microstructure (two-scale) model
\begin{eqnarray}\label{upsceq}
&&\begin{cases}
\partial_{t}v_{0}(x,y,t)=D_l \Delta_y v_0(x,y,t)& \mbox{for}\,\,\,y\in B(x),\,\,x\in\Omega,\\
\partial_{t}\left(\theta(x)u_{0}+\int_{|y|<r(x)} v_0\,dy\right)= \\
\hspace{4cm}\nabla_{x}\cdot
({\cal D}(x)\nabla_{x}u_{0}-\bar{q}u_0)&\mbox{for}\,\,\,x\in
\Omega,
\end{cases} \label{upsc1}\\
&&\begin{cases}
v_0(x,y,t)=u_0(x,t)&\mbox{for}\,\,\,y\in\partial B(x),\\
u_0(x,t)=u_b(x,t)&\mbox{for}\,\,\,x\in\Gamma,
\end{cases}\label{upsc2}\\
&&\begin{cases}
u_0(x,0)=u_I(x)&\mbox{for}\,\,\,x\in
\Omega,\\
v_0(x,y,0)=v_I(x,y)& \mbox{for}\,\,\,y\in B(x),\,\,x\in\Omega.
\end{cases}\label{upsc3}
\end{eqnarray}
where the porosity $\theta(x)$ of the medium is given by
\begin{align*}
\theta(x):=|Y(x)|=1-|B(x)|,
\end{align*}
and where the effective diffusivity tensor ${\cal D}(x)$ is defined by
\begin{eqnarray*}
{\cal D}(x):=D_h\int_{Y(x)}(I+\nabla_yM(x,y))\,dy,
\end{eqnarray*}
with $M=(M_1,M_2)$.



\section{Main result}\label{main_result}

Recall that  $(u_\epsilon, v_\epsilon)$ is the  solution vector for the micro problem and $(u_0, v_0)$ is the solution vector for the two-scale problem. We introduce now the macroscopic reconstructions\footnote{We borrowed this terminology from \cite{Eck_correctors}. Note however that the concept of reconstruction operators appears in various other frameworks like for the heterogeneous multiscale method \cite{HMM}.} $u_0^\epsilon, v_0^\epsilon, u_1^\epsilon$,
which are defined as follows
\begin{eqnarray}
u_0^\epsilon(x,t)&:=&u_0(x,t) \mbox{ for all } x\in \Omega_h^\epsilon,\\
v_0^\epsilon(x,t)&:=&v_0(x,x/\epsilon,t) \mbox{ for all } x\in \Omega_l^\epsilon,\\
u_1^\epsilon(x,t)&:=&u_0^\epsilon(x,t)+\epsilon M(x,x/\epsilon)\nabla u_0^\epsilon(x,t) \mbox{ for all } x\in \Omega_h^\epsilon,
\end{eqnarray}
In the same spirit, we introduce the reconstructed flow velocity $q^\epsilon_0(x)=q(x)$ for all  $x\in \Omega_h^\epsilon$ and the corresponding reconstructions $u_{0I}^\epsilon,v_{0I}^\epsilon$ for the macroscopic initial data $u_I$ and $v_I$, respectively.

The main result of our paper is stated in the next Theorem. The
applicability of this result confines to our working assumptions
(A1), (A2), (B1), (B2), and (B3) that we introduce in Section
\ref{technical}.

\begin{theorem}\label{MR}
Assume (A1), (A2), (B1), (B2), and (B3). Then the following
convergence rate holds
\begin{eqnarray}
||u_\epsilon-u_0^\epsilon||_{L^\infty(I, L^2(\Omega_h^\epsilon))}+||v_\epsilon-v_0^\epsilon||_{L^\infty(I, L^2(\Omega_l^\epsilon))}+\nonumber\\
||u_\epsilon-u_1^\epsilon||_{L^\infty(I, H^1(\Omega_h^\epsilon))}+\epsilon||v_\epsilon-v_0^\epsilon||_{L^\infty(I, H^1(\Omega_l^\epsilon))}\leq c\sqrt{\epsilon},
\end{eqnarray}
where $I=(0,T]$, and where the constant $c$ is independent of $\epsilon$.
\end{theorem}

The remainder of the paper is concerned with the proof of Theorem \ref{MR}.

\section{Technical preliminaries}\label{technical}

\subsection{Function Spaces. Assumptions. Known results}

\subsubsection{Functional setting}

As space of test functions for the microscopic problem, we take the Sobolev space
$$H^1(\Omega_\epsilon;\partial\Omega):=\left\{\varphi\in H^1(\Omega_\epsilon) \mbox{ with } \varphi=0 \mbox{ at }\partial\Omega\right\}.$$
Since the formulation of the upscaled problem involves two distinct spatial variables $x\in\Omega$ and $y\in B(x)$ (with $B(x)\subset \Omega$),  we need to introduce the following spaces:
\begin{align}
&V_1:=H_0^1(\Omega),\\
&V_2:=L^2(\Omega;H^2(B(x))),\\
&H_1:=L^2_\theta(\Omega),\\
&H_2:=L^2(\Omega;L^2(B(x))).
\end{align}
The spaces $H_2$ and $V_2$ make sense, for instance, as indicated in \cite{sebam_equadiff}.

\subsubsection{Assumptions for the microscopic model}

{\bf Assumption (A1)}:  We assume the following restrictions on data and parameters:  Take $D_h,D_l\in ]0,\infty[$, $u_\epsilon^I\in H^1(\Omega_h^\epsilon)$, $v_\epsilon^I\in H^1(\Omega_l^\epsilon)$, and $u_b\in H^1(I;H^3(\Gamma))$. Furthermore, we assume
\begin{eqnarray*}
||u_\epsilon^I-u_{0I}^\epsilon||_{L^\infty(\Omega_h^\epsilon)}=\mathcal{O}(\epsilon^{\theta_1}) \mbox{ with } \theta_1\geq \frac{1}{2} \\
||v_\epsilon^I-v_{0I}^\epsilon||_{L^\infty(\Omega_l^\epsilon)}=\mathcal{O}(\epsilon^{\theta_2}) \mbox{ with } \theta_2\geq \frac{1}{2}
\end{eqnarray*}
\newline
{\bf Assumption (A2)}: $q_\epsilon\in H^2(\Omega_h^\epsilon;\mathbb{R}^d)\cap L^\infty(\Omega_h^\epsilon;\mathbb{R}^d)$ with $\nabla\cdot q_\epsilon=0$ a.e. in $\Omega_h^\epsilon$. Additionally, we assume that
\begin{equation}\label{st2}
||q_\epsilon-\frac{1}{\theta}\bar{q}||_{L^2(\Omega_h^\epsilon;\mathbb{R}^d)}=\mathcal{O}(\epsilon^{\theta_3})
\mbox{ with } \theta_3\geq \frac{1}{2}.\end{equation}

 \begin{remark} If one wishes to replace $q_\epsilon$ with the stationary Stokes or Navier-Stokes equations, then a few additional things have to be taken into account. One of the most striking facts is that the exponent
 $\theta$ arising in (\ref{st2}) seems to be restricted to $\theta_3=\frac{1}{6}$.  Consequently, this worsens essentially the convergence rate;
 see, for instance, \cite{Marusic} (Theorem 1) for a discussion of homogenization of the periodic case.
 \end{remark}

\subsubsection{Assumptions for the two-scale model}

{\bf Assumption (B1)}:
The level set function $S:\Omega\times{\cal U}\rightarrow \mathbb{R}$ is 1-periodic in its second variable and is in
$C^2(\Omega\times{\cal U})$.

{\bf Assumption (B2)}:  We assume the following restrictions on data and parameters:
\begin{align*}
\begin{cases}
\theta,\, D \in L^\infty_+(\Omega)\cap H^2(\Omega),&\\
\bar{q}\in H^2(\Omega;\mathbb{R}^d)\cap L^\infty(\Omega;\mathbb{R}^d)\,\,\mbox{with}\,\, \nabla\cdot \bar{q} =0 \mbox{ a.e. in } \Omega,&\\
u_b\in L_+^\infty(\Omega\times I)\cap H^1(I;H^3(\Gamma)),&\\
\partial_tu_b\leq 0\,\, \mbox{a.e.}\,\,(x,t)\in\Omega\times I,&\\
u_I\in L_+^\infty(\overline{\Omega})\cap H_1\cap H^2(\Omega),&\\
v_I(x,\cdot)\in L_+^\infty(B(x))\cap H_2\,\,\mbox{for a.e.}\,\,x\in\overline{\Omega}.&
\end{cases}
\end{align*}

{\bf Assumption (B3)}:
\begin{align}
&H_1:=H_0^1(\Omega),\\
&H_2:=L^2(\Omega;H^2(B(x))),\\
&V_1:=H^2_\theta(\Omega),\\
&V_2:=L^2(\Omega;H^3(B(x))).
\end{align}

\begin{remark}Following the lines of \cite{sebam_equadiff} and \cite{Show_walk}, Assumption (B1)  implies in particular that the measures $|\partial B(x)|$ and $|B(x)|$ are bounded away from zero (uniformly in $x$). Consequently, the following direct Hilbert integrals (cf. \cite{Dix} (part II, chapter 2), e.g.)
\begin{eqnarray}
L^2(\Omega;H^1(B(x)))&:=&\{u\in L^2(\Omega;L^2(B(x))): \nabla_y u\in L^2(\Omega;L^2(B(x)))\}\nonumber\\
L^2(\Omega;H^1(\partial B(x)))&:=& \{u:\Omega\times \partial B(x)\to \mathbb{R} \mbox{ measurable  such that }  \int_\Omega ||u(x)||^2_{L^2(\partial B(x))}<\infty\} \nonumber
\end{eqnarray}
are well-defined separable Hilbert spaces and, additionally, the {\em distributed trace}
$$\gamma: L^2(\Omega;H^1(B(x)))\to L^2(\Omega, L^2(\partial B(x)))$$ given by
\begin{equation}\label{g}
\gamma u(x,s):=(\gamma_x U(x))(s), \ x\in \Omega, s\in \partial B(x),u\in L^2(\Omega;H^1(B(x)))
\end{equation}
is a bounded linear operator.  For each fixed $x\in \Omega$, the map $\gamma_x$, which is arising in (\ref{g}), is the standard trace operator from $H^1(B(x))$ to $L^2(\partial B(x))$. We refer the reader to \cite{sebam_PhD} for more details on the construction of these spaces and to \cite{sf} for the definitions of their duals as well as for  a less regular condition (compared to (B1)) allowing to define these spaces in the context of a certain class of anisotropic Sobolev spaces.
\end{remark}

For convenience, we also introduce the evolution triple $(\mathbb{V},\mathbb{H},\mathbb{V}^*)$,
where
\begin{align}
&\mathbb{V}:=\{(\phi,\psi)\in V_1\times V_2\,|\,\phi(x)=\psi(x,y)\, \mbox{for}\,
x\in \Omega,\, y\in \partial B(x)\},\\
&\mathbb{H}:=H_1\times H_2,
\end{align}

\subsubsection{Analysis of microscopic equations}\label{analysis_micro}

\begin{definition}\label{def_weak_micro}
Assume (A1), (A2) and (B1). The pair $(u_\epsilon,v_\epsilon)$, with $u_\epsilon=U_\epsilon+u_b$ and
where $(U_\epsilon,v_\epsilon)\in H^2(\Omega_h^\epsilon;\partial \Omega)\times H^2(\Omega_l^\epsilon)$, is a weak solution of the problem $(P_\epsilon)$
if the transmission conditions (\ref{transmission}) are fulfilled and the following identities hold
\begin{align}
&\int_{\Omega_h^\epsilon}\partial_t(U_\epsilon+u_b)\phi\, dx+\int_{\Omega_h^\epsilon}
(D_h\nabla (U_\epsilon+u_b)-q_\epsilon(U_\epsilon+u_b))\cdot \nabla \phi \, dx=\nonumber\\
&\hspace{8cm}-\int_{\Gamma_\epsilon} \nu^\epsilon\cdot (\epsilon^2D_l\nabla v_\epsilon)\phi ds,\\
&\int_{\Omega_l^\epsilon} \partial_t v_\epsilon \psi \,dydx +
\int_{\Omega_l^\epsilon}  \epsilon^2 D_l \nabla v_\epsilon\cdot \nabla \psi \,dx =
\int_{\partial \Gamma_\epsilon}\nu^\epsilon \cdot (\epsilon^2D_l\nabla v_\epsilon)\psi \, ds
dx,
\end{align}
for all $(\phi,\psi)\in H^1(\Omega_h^\epsilon;\partial \Omega)\times H^1(\Omega_l^\epsilon)$ and $t\in I$.
\end{definition}
\begin{theorem}\label{main_result_micro} Assume (A1), (A2) and (B1).
Problem ($P_\epsilon$) admits  a unique global-in-time weak solution in the
sense of Definition \ref{def_weak_micro}.
\end{theorem}
\begin{proof} Since we deal here with a linear transmission problem, the proof of the Theorem can be done with  standard techniques  (see \cite{Evans}, e.g.).
\end{proof}

\subsubsection{Analysis of two-scale equations}\label{analysis_macro}

This section contains basic results concerning the well-posedness of the two-scale problem, which we reformulate here as:
\begin{align*}
(P)\begin{cases}
\theta(x)\partial_tu_0-\nabla_x\cdot({\cal D}(x)\nabla_x u_0 -q u_0)=-\int_{\partial B(x)}
\nu_0\cdot(D_l\nabla_y v_0)\,d\sigma & \mbox{in}\,\, \Omega,\\
\partial_tv_0-D_l \Delta_y v_0=0 & \mbox{in}\,\, B(x)
,\\
u_0(x,t)=v_0(x,y,t)& \mbox{at}\,\, (x,y)\in \Omega\times \partial B(x),\\
u_0(x,t)=u_b(x,t) & \mbox{at}\,\, x\in \partial \Omega,\\
u_0(x,0)=u_I(x) & \mbox{in}\,\,\overline{\Omega},\\
v_0(x,y,0)=v_I(x,y) & \mbox{at}\,\,(x,y)\in \overline{\Omega}\times\overline{B(x)}.
\end{cases}
\end{align*}

Before starting to discuss the existence and uniqueness of weak solutions to problem $(P)$, we denote $U:=u-u_b$ and notice that $U=0$ at $\partial \Omega$.

\begin{definition}\label{def_weak}
Assume (B1) and (B2). The pair $(u,v)$, with $u=U+u_b$
where $(U,v)\in\mathbb{V}$, is a weak solution of the problem $(P)$
if the following identities hold
\begin{align}
&\int_\Omega \theta \partial_t(U+u_b)\phi\, dx+\int_\Omega
(D\nabla_x(U+u_b)-q(U+u_b))\cdot \nabla_x \phi \, dx=\nonumber\\
&\hspace{8cm}-\int_\Omega \int_{\partial B(x)}\nu_0\cdot (D_l\nabla_y v)\phi \, d\sigma
dx,\\
&\int_{\Omega} \int _{B(x)}\partial_t v \psi \,dydx +
\int_{\Omega} \int _{B(x)} D_l \nabla_y \cdot \nabla_y \psi \,dydx =
\int_\Omega \int_{\partial B(x)}\nu_0\cdot (D_l\nabla_y v)\phi \, d\sigma
dx,
\end{align}
for all $(\phi,\psi)\in \mathbb{V}$ and $t\in I$.
\end{definition}

\begin{proposition}[Uniqueness] Assume (B1) and (B2). Problem ($P$) admits at most one weak solution in the sense of Definition \ref{def_weak}.
\end{proposition}
\begin{proof}
Since Problem ($P$) is linear, the uniqueness follows in the
standard way and can be done directly in the $x$-dependent function
spaces: One takes two different weak solutions to ($P$) satisfying
the same initial data. Testing with their difference and using the
Gronwall's inequality conclude the proof.
\end{proof}
\begin{theorem}\label{main_result_macro}
Assume (B1) and (B2). Problem (P) admits at least a global-in-time weak solution in the
sense of Definition \ref{def_weak}.
\end{theorem}
\begin{proof} See the proof of Theorem 5.11 in \cite{Tycho_Adrian}.
\end{proof}
To get the correctors estimates stated in Theorem \ref{MR} we need more two-scale regularity for the macroscopic reconstruction of the concentration field $v_0$. We state this fact in the following result:

\begin{lemma}[Additional two-scale regularity] \label{addreg}
Assume (B1) and (B2). Then
\begin{equation}
v_0^\epsilon\in L^2(I;H^2(\Omega;H^2(B(x))))).
\end{equation}
\end{lemma}
\begin{proof}
The proof of this result is similar to the regularity lift proven in Claim 5.10 in \cite{Tycho_Adrian}. The main ingredients are fixing of the boundary and testing with difference quotients  in the weak formulation posed in fixed domains. We omit to show the details.
\end{proof}
\begin{theorem}[On strong solutions] \label{strongsol}
Assume (B1), (B2), and (B3). Problem (P) admits one  global-in-time strong solution.
\end{theorem}
\begin{proof}
Under the assumptions (B1) and (B2), Theorem \ref{main_result_macro} guarantees the existence of global-in-time weak solutions. Relying on the assumption (B3), we can lift the regularity until getting a strong solution. A similar calculation is done in  \cite{Evans} (Theorem 5, pp. 360--364). In particular, (B3) allows for a regularity lift such that
$$||\nabla \partial_t u_0||_{L^2(I\times\Omega)}+||\nabla \partial_t v_0||_{L^2(I\times\Omega\times Y)}\leq c.$$
For our purpose, we only need
\begin{eqnarray}
||\nabla \partial_t u_0^\epsilon||_{L^2(I\times\Omega_h^\epsilon)}\leq c,\\
||u_0^\epsilon||_{L^2(I;H^3(\Omega_h^\epsilon))}\leq c\label{regi}
\end{eqnarray}
where $u_0^\epsilon$ is the macroscopic reconstruction of $u_0$. Quite probably (\ref{regi}) could be relaxed to $||u_0^\epsilon||_{L^2(I;H^{2+\theta}(\Omega_h^\epsilon))}$ for some $\theta>0$, but we don't address here this issue.  Since we rather wish that the reader focusses on our strategy of getting the correctors,  we omit to show the proof details for this regularity result.
\end{proof}

\section{Proof of Theorem \ref{MR}}\label{Proofs}
In this section, we give the proof of the main result of our paper,
i.e. of Theorem \ref{MR}. The proof uses the auxiliary results
stated in the lemmas below. They mainly concern integral estimates for rapidly oscillating functions with prescribed average; Related estimates can be found, for instance, in  \cite{Chechkin} and Section 1.5 in \cite{CPS}.
\begin{lemma}\label{lem1} Assume the hypothesis of Theorem \ref{MR}
to hold. Then
\begin{align*}
\int_{Y(x)}\Big(\nabla_x\cdot((I+\nabla_yM)\nabla_xu_0)-\frac{1}{\theta(x)}\nabla_x\cdot\int_{Y(x)}(I+\nabla_yM)\nabla_xu_0\,dy\Big)\,dy=\\
-\int_{\partial B(x)}\nu_1\cdot(I+\nabla_yM)\nabla_xu_0\,d\sigma.
\end{align*}
\end{lemma}
\begin{proof}
We compute
\begin{align*}
\int_{Y(x)}\Big(\nabla_x\cdot((I+\nabla_yM)&\nabla_xu_0)-\frac{1}{\theta(x)}\nabla_x\cdot\int_{Y(x)}(I+\nabla_yM)\nabla_xu_0\,dy\Big)\,dy=\\
&\int_{Y(x)}\nabla_x\cdot((I+\nabla_yM)\nabla_xu_0\,dy-\nabla_x\cdot\int_{Y(x)}(I+\nabla_yM)\nabla_xu_0\,dy.
\end{align*}
Reynolds's transport theorem (see for instance \cite{Gurtin}) gives
\begin{align*}
\nabla_x\cdot\int_{Y(x)}(I+\nabla_yM)\nabla_xu_0\,dy=\int_{Y(x)}\nabla_x\cdot((I+\nabla_yM)\nabla_xu_0)\,dy\\
+\int_{\partial B(x)}\frac{\nabla_x S}{|\nabla_yS|}(I+\nabla_y M)\nabla_xu_0\,d\sigma.
\end{align*}
By the boundary condition in \eqref{cellpu}, we have
\begin{align*}
\nu_0(I+\nabla_yM)=\frac{\nabla_yS}{|\nabla_yS|}(I+\nabla_yM)=0\,\,\,\mbox{on}\,\,\,\partial B(x),
\end{align*}
so that we can write, using \eqref{nu11},
\begin{align*}
\frac{\nabla_x S}{|\nabla_yS|}(I+\nabla_y M)\nabla_xu_0&=
\Bigg(
\frac{\nabla_x S}{|\nabla_yS|}-
\frac{\nabla_xS\cdot\nabla_yS}{|\nabla_yS|^2}
\frac{\nabla_yS}{|\nabla_yS|}
\Bigg)(I+\nabla_y M)\nabla_xu_0\\
&=\nu_1\cdot(I+\nabla_yM)\nabla_xu_0
\end{align*}
on $\partial B(x)$.
Combining the above expressions proves the conclusion of this lemma.
\end{proof}

\begin{lemma} \label{lem2} Assume the hypothesis of Theorem \ref{MR}
to hold. Let $Q(x,y)\in L^2(\Omega;L^2(B(x)))$ and $p\in
 L^2(\Omega;L^2(\partial B(x)))$. Furthermore, suppose that
$\int_{Y(x)}Q(x,y)\,dy=\int_{\partial Y(x)}p(x,y)\,d\sigma$. Then
the inequality
\begin{align*}
\left|\int_{\Omega^\epsilon_h}Q(x,x/\epsilon)\phi(x)\, dx-\epsilon\int_{\Gamma^\epsilon} p(x,x/\epsilon)\phi(x)\, ds \right|\leq C\epsilon\|\phi\|_{H^1(\Omega^\epsilon_h)}
\end{align*}
holds for every $\phi\in H^1(\Omega^\epsilon_h;\partial\Omega)$. The
constant $C$ does not depend on the choice of $\epsilon$.
\end{lemma}
\begin{proof}
The problem
\begin{align*}
&\Delta_y\Psi(x,y)=Q(x,y)\,\,\,\,\mbox{in}\,\,\,Y(x) \\
&\nu_0\cdot{\nabla_y}\Psi=p(x,y)\,\,\,\,\mbox{on} \,\,\,\partial B(x)
\end{align*}
has a 1-periodic in $y$ solution that is unique up to an additive constant.
We multiply the first equation above with $\phi$ and integrate over $\Omega^\epsilon_h$ and use the thus obtained equality to get:
\begin{align*}
&\left|\int_{\Omega^\epsilon_h}Q(x,x/\epsilon)\phi(x)\, dx-\epsilon\int_{\Gamma^\epsilon} p(x,x/\epsilon)\phi(x)\, ds
\right|=\\
&\left|\int_{\Omega^\epsilon_h}\Delta_y\Psi(x,y)|_{y=\frac{x}{\epsilon}}\phi(x)\,dx-\epsilon\int_{\Gamma^\epsilon} p(x,x/\epsilon)\, ds\right|=\\
&\left|\epsilon\int_{\Omega^\epsilon_h}\Big(\nabla_x[\nabla_y\Psi(x,y)|_{y=\frac{x}{\epsilon}}]-\nabla_x\nabla_y\Psi(x,y)|_{y=\frac{x}{\epsilon}}\Big)\phi(x)\,dx-\epsilon\int_{\Gamma^\epsilon} p(x,x/\epsilon)\, ds\right|=\\
&\left|\epsilon\int_{\Gamma^\epsilon}(\nu_0+\epsilon\nu^\epsilon_1)\cdot\nabla_y\Psi(a,y)|_{y=\frac{x}{\epsilon}}\phi(x)\,ds
-\epsilon\int_{\Omega^\epsilon_h}\nabla_y\Psi|_{y=\frac{x}{\epsilon}}\nabla_x\phi(x)\,dx\right.\\
&\,\,\,\,\,\,\,\left.-\epsilon\int_{\Omega^\epsilon_h}\nabla_x\nabla_y\Psi(x,y)|_{y=\frac{x}{\epsilon}}\phi(x)\,dx-\epsilon\int_{\Gamma^\epsilon} p(x,x/\epsilon)\, ds\right|\leq\\
&\epsilon^2\left|\int_{\Gamma^\epsilon}\nu^\epsilon_1\cdot\nabla_y\Psi(a,y)|_{y=\frac{x}{\epsilon}}\phi(x)\,ds\right|+\epsilon\left|\int_{\Omega^\epsilon_h}\nabla_y\Psi|_{y=\frac{x}{\epsilon}}\nabla_x\phi(x)\,dx\right|\\
&\,\,\,\,\,\,\,+\epsilon\left|\int_{\Omega^\epsilon_h}\nabla_x\nabla_y\Psi(x,y)|_{y=\frac{x}{\epsilon}}\phi(x)\,dx\right|\leq\epsilon C
\|\phi\|_{H^1(\Omega^\epsilon_h)}
\end{align*}
The lemma is now proved.
\end{proof}

The last auxiliary lemma is a special case of Lemma 4 in \cite{Chechkin}, and therefore we will state it here without proof.

\begin{lemma}\label{lemch}
Let $\Pi_\epsilon$ be a subset of $\{x\in\Omega\,|\,\mbox{dist}(x,\partial\Omega)\leq
\frac{\sqrt{2}}{2}\epsilon\}$. Then the following inequality
\begin{align*}
\left|\int_{\Pi_\epsilon}\nabla_xu_0\phi\,dx\right|\leq C \epsilon^{3/2}
\|\phi\|_{H^1(\Omega_h^\epsilon)}
\end{align*}
holds for all $\phi\in H^1(\Omega_h^\epsilon;\partial\Omega)$. The constant $C$ does not depend on $\epsilon$.
\end{lemma}

\begin{proof}(of Theorem \ref{MR})
We define
\begin{align*}
&z_\epsilon(x,t):=u_0(x,t)+\epsilon \chi_\epsilon(x) M(x,x/\epsilon)\nabla u_0(x,t)-u_\epsilon(x,t),\\
&w_\epsilon(x,t):=v_0(x,x/\epsilon,t)-v_\epsilon(x,t).
\end{align*}
By, e.g., Theorem 4 in \cite{Lukkassen}, the functions $M(x,x/\epsilon)$ and $v_0(x,x/\epsilon,t)$ are well-defined functions in $H^1(\Omega_h^\epsilon)$ and $L^2(I;H^1(\Omega^\epsilon_l))$, respectively, and furthermore, we know
by construction that there exists a $C>0$ such that
\begin{eqnarray}
&\|z_\epsilon\|_{L^2(I;H^1(\Omega^\epsilon_h))}\leq C,\label{EB1}\\
&\|w_\epsilon\|_{L^2(I;H^1(\Omega^\epsilon_l))}\leq \epsilon
C,\label{EB2}
\end{eqnarray}
where the constant $C$ is independent of the choice of $\epsilon$.
In the following we will use the notation $u_1(x,y,t)=M(x,y)\nabla u_0(x,t)$.
We compute:
\begin{align*}
\Delta z_\epsilon(x,t)=&\Delta u_0(x,t)+\epsilon \chi_\epsilon \Delta_xu_1(x,y,t)|_{y=\frac{x}{\epsilon}} + 2 \chi_\epsilon \nabla_x\cdot\nabla_y u_1(x,y,t)|_{y=\frac{x}{\epsilon}} \\
&
+ \frac{1}{\epsilon}\chi_\epsilon \Delta_y
u_1(x,y,t)|_{y=\frac{x}{\epsilon}}+\epsilon\Delta\chi_\epsilon u_1(x,y,t)|_{y=\frac{x}{\epsilon}}-\Delta u_\epsilon(x,t)\\
&+2\epsilon \nabla \chi_\epsilon \cdot \nabla_x u_1(x,y,t)|
_{y=\frac{x}{\epsilon}}
+ 2\nabla \chi_\epsilon \cdot \nabla_y u_1(x,y,t)|_{y=\frac{x}{\epsilon}}  \\
\Delta_x w_\epsilon(x,\dfrac{x}{\epsilon},t)=&\Delta_xv_0(x,t)
+ \epsilon^{-1}\nabla_x\cdot\nabla_y v_0(x,y,t)|_{y=\frac{x}{\epsilon}}\\
&+\epsilon^{-1}\nabla_y\cdot\nabla_x v_0(x,y,t)|_{y=\frac{x}{\epsilon}}+ \frac{1}{\epsilon^2}\Delta_y
v_0(x,y,t)|_{y=\frac{x}{\epsilon}}-\Delta_xv_\epsilon(x,y,t)|_{y=\frac{x}{\epsilon}}.
\end{align*}
We use that
\begin{align*}
\theta(x)\partial_tu_0-\nabla_x\cdot({\cal D}(x)\nabla_x u_0 -\bar{q}u_0)=-\int_{\partial B(x)}
\nu_0\cdot(D_l\nabla_y v_0)\,d\sigma & \mbox{in}\,\, \Omega
\end{align*}
and $\Delta_y u_1(x,y)=0$ and $\partial_t u_\epsilon = \nabla\cdot(D_h \nabla u_\epsilon-q_\epsilon u_\epsilon)$ to obtain
\begin{align*}
D_h\Delta_x z_\epsilon-\partial_t z_\epsilon=
&D_h \Delta_xu_0
+ \epsilon \chi_\epsilon D_h\Delta_x u_1
+2\chi_\epsilon D_h\nabla_x\cdot\nabla_y u_1
+\epsilon D_h \Delta\chi_\epsilon u_1\\
&+2\epsilon D_h\nabla \chi_\epsilon \cdot \nabla_x u_1
+ 2 D_h \nabla \chi_\epsilon \cdot \nabla_y u_1\\
&-\frac{1}{\theta}\Big(\nabla_x\cdot({\cal D}(x)\nabla_x u_0-\bar{q}u_0)-\int_{\partial B(x)}
\nu_0\cdot(D_l\nabla_yv_0)\,d\sigma\Big)-\epsilon\chi_\epsilon\partial_t u_1\\
&-\nabla\cdot(q_\epsilon u_\epsilon),\\
\epsilon^2D_l\Delta_x w_\epsilon-\partial_t w_\epsilon=&\epsilon^2D_l \Delta_xv_0 +\epsilon D_l(\nabla_x\cdot\nabla_y v_0+\nabla_y\cdot\nabla_x v_0).
\end{align*}
On $\Gamma^\epsilon$ we have
\begin{align*}
\nu^\epsilon\cdot \nabla z_\epsilon=&-\nu^\epsilon\cdot \nabla u_\epsilon+\nu^\epsilon\cdot\nabla_x u_0+\epsilon \nu^\epsilon\cdot \nabla_x u_1+\nu^\epsilon\cdot \nabla_y u_1\\
=&-\epsilon^2\frac{D_l}{D_h}\nu^\epsilon\cdot\nabla v_\epsilon
+\nu^\epsilon\cdot\nabla_x u_0+\epsilon \nu^\epsilon\cdot \nabla_x u_1+\nu^\epsilon\cdot \nabla_y u_1,\\
\nu^\epsilon\cdot \nabla w_\epsilon=&-\nu^\epsilon\cdot \nabla v_\epsilon+\nu^\epsilon\cdot\nabla_x v_0+\epsilon^{-1}\nu^\epsilon\cdot \nabla_y v_0.
\end{align*}
Now, we multiply with $\phi$ and integrate by parts to get
\begin{eqnarray}\label{weakz}
\int_{\Omega^\epsilon_h}\partial_t z_\epsilon \phi \,dx&+&
\int_{\Omega^\epsilon_h}(D_h\nabla z_\epsilon \cdot \nabla \phi+q^\epsilon\cdot\nabla z_\epsilon \phi) \,dx=\epsilon\int_{\Omega^\epsilon_h}\chi_\epsilon \partial_t u_1\phi\,dx
-\int_{\Omega^\epsilon_h}D_h \Delta_xu_0\phi\,dx\nonumber\\
&-&\epsilon\int_{\Omega^\epsilon_h}D_h\chi_\epsilon \Delta_x u_1\phi\,dx
-2\int_{\Omega^\epsilon_h} D_h \chi_\epsilon \nabla_x\cdot\nabla_y u_1 \phi\,dx\nonumber\\
&+&\int_{\Omega^\epsilon_h}\frac{1}{\theta}\Big(\nabla_x\cdot({\cal D}(x)\nabla_x u_0)-\int_{\partial B(x)}
\nu_0\cdot(D_l\nabla_yv_0)\,d\sigma\Big)\phi\,dx \\
&+&\epsilon^2\int_{\Gamma^\epsilon}D_l\nu^\epsilon\cdot\nabla v_\epsilon \phi\,ds
-D_h\int_{\Gamma^\epsilon}\nu^\epsilon\cdot\nabla_x u_0\phi\,ds\nonumber\\
&-&\epsilon D_h\int_{\Gamma^\epsilon}\nu^\epsilon\cdot \nabla_x u_1\phi\,ds-D_h\int_{\Gamma^\epsilon}\nu^\epsilon\cdot \nabla_y u_1\phi\,ds\nonumber \\
&-&\int_{\Omega_h^\epsilon}(\epsilon D_h \Delta\chi_\epsilon u_1+2\epsilon D_h\nabla \chi_\epsilon \cdot \nabla_x u_1
+2 D_h \nabla \chi_\epsilon \cdot \nabla_y u_1)\phi\,dx\nonumber\\
&-&\int_{\Omega_h^\epsilon}(\frac{1}{\theta}\bar{q}-q_\epsilon)\cdot\nabla u_0\phi\,dx
-\epsilon\int_{\Omega_h^\epsilon}\chi_\epsilon u_1 q_\epsilon \cdot \nabla \phi\,dx,
\end{eqnarray}
and
\begin{eqnarray}\label{weakw}
\int_{\Omega^\epsilon_l}\partial_t w_\epsilon \psi \,dx&+&
\epsilon^2 \int_{\Omega^\epsilon_l}D_l\nabla w_\epsilon \cdot \nabla \psi \,dx=-\epsilon^2\int_{\Omega^\epsilon_l}D_l \Delta_xv_0\psi\,dx\nonumber\\
&-&\epsilon\int_{\Omega^\epsilon_l} D_l(\nabla_x\cdot\nabla_y v_0+\nabla_y\cdot\nabla_x v_0) \psi\,dx
-\epsilon^2\int_{\Gamma^\epsilon}D_l\nu^\epsilon\cdot\nabla v_\epsilon \psi\,ds\nonumber\\
&+&\epsilon^2\int_{\Gamma^\epsilon}D_l\nu^\epsilon\cdot\nabla_x v_0 \psi\,ds+\epsilon\int_{\Gamma^\epsilon}D_l\nu^\epsilon\cdot\nabla_y v_0 \psi\,ds.
\end{eqnarray}
We take into account the identity
\begin{align*}
\big(\nabla_y\cdot \nabla_x u_1(x,y)\big)|_{y=x/\epsilon}=
\epsilon \nabla_x\cdot(\nabla_x u_1(x,x/\epsilon))-\epsilon\big(\Delta_x u_1(x,y)\big)|_{y=x/\epsilon},
\end{align*}
which gives
\begin{align*}
\epsilon D_h\int_{\Gamma^\epsilon}\nu^\epsilon\cdot\nabla_xu_1|_{y=x/\epsilon} z_\epsilon\,ds=&\epsilon D_h\int_{\Omega_h^\epsilon} \nabla_x u_1|_{y=x/\epsilon} \cdot \nabla (\chi_\epsilon z_\epsilon)
+ \chi_\epsilon z_\epsilon \nabla_x\cdot(\nabla_xu_1|_{y=x/\epsilon}) \,dx\\
=&\epsilon D_h\int_{\Omega_h^\epsilon} \nabla_x u_1 \cdot \nabla
(\chi_\epsilon z_\epsilon) +\chi_\epsilon z_\epsilon\Delta_xu_1 \,dx\\
&+D_h\int_{\Omega_h^\epsilon}\chi_\epsilon z_\epsilon \nabla_y\cdot\nabla_xu_1\, dx,
\end{align*}
and also the boundary condition $\nu_0 \cdot \nabla_yu_1 =-\nu_0\cdot \nabla_xu_0$ for $y\in \partial B(x)$, which holds also on $\Gamma^\epsilon$, we add the two equations \eqref{weakz} and \eqref{weakw}, substitute $\phi=z_\epsilon$ and $\psi=w_\epsilon$ and obtain
\begin{align*}
&\frac{1}{2}\partial_t\|z_\epsilon\|^2_{L^2(\Omega_h^\epsilon)}+\frac{1}{2}\partial_t\|w_\epsilon\|^2_{L^2(\Omega_l^\epsilon)}
+\int_{\Omega_h^\epsilon}(D_h\nabla z_\epsilon\cdot \nabla z_\epsilon+q^\epsilon\cdot\nabla z_\epsilon)\,dx
+\epsilon^2
D_l\|\nabla w_\epsilon\|^2_{L^2(\Omega_l^\epsilon)} =\\
&-\int_{\Omega^\epsilon_h}D_h \Delta_xu_0z_\epsilon\,dx
-\int_{\Omega^\epsilon_h} \chi_\epsilon D_h(\nabla_x\cdot\nabla_y u_1) z_\epsilon \,dx+\int_{\Omega^\epsilon_h}\frac{1}{\theta}\nabla_x\cdot({\cal D}(x)\nabla_x u_0)z_\epsilon\,dx\\
&-\epsilon D_h\int_{\Gamma^\epsilon}\nu^\epsilon_1\cdot(\nabla_x u_0+\nabla_yu_1)z_\epsilon\,ds\\
&-\int_{\Omega^\epsilon_h}(\int_{\partial B(x)}
\nu_0\cdot(D_l\nabla_yv_0)\,d\sigma)z_\epsilon\,dx+\epsilon\int_{\Gamma^\epsilon}D_l\nu^\epsilon\cdot\nabla_y v_0 z_\epsilon\,ds\\
&+\epsilon\int_{\Omega^\epsilon_h}\chi_\epsilon \partial_t u_1z_\epsilon\,dx-\epsilon\int_{\Omega^\epsilon_h}D_h\nabla_x u_1 \cdot\nabla(\chi_\epsilon z_\epsilon)\,dx\\
&-\epsilon^2\int_{\Omega^\epsilon_l}D_l \Delta_xv_0w_\epsilon\,dx
-\epsilon\int_{\Omega^\epsilon_l} D_l(\nabla_x\cdot\nabla_y v_0+\nabla_y\cdot\nabla_x v_0) w_\epsilon\,dx
+\epsilon^2\int_{\Gamma^\epsilon}D_l\nu^\epsilon\cdot\nabla_x v_0 w_\epsilon\,ds\\
&+\epsilon^2\int_{\Gamma^\epsilon}D_l\nu^\epsilon\cdot(\epsilon\nabla v_\epsilon-\nabla_yv_0)u_1\,ds -\int_{\Omega_h^\epsilon}(\epsilon D_h \Delta\chi_\epsilon u_1+2\epsilon D_h\nabla \chi_\epsilon \cdot \nabla_x u_1
+2 D_h \nabla \chi_\epsilon \cdot \nabla_y u_1)z_\epsilon \,dx\\
&-\int_{\Omega_h^\epsilon}(\frac{1}{\theta}\bar{q}-q_\epsilon)\cdot\nabla u_0z_\epsilon\,dx
-\epsilon\int_{\Omega_h^\epsilon}\chi_\epsilon u_1 q_\epsilon \cdot \nabla z_\epsilon\,dx,
\end{align*}
where we have also used that $\phi-\psi=z_\epsilon-w_\epsilon=\epsilon u_1$ on $\Gamma^\epsilon$ and $\phi=\epsilon u_1$ on $\partial\Omega$.

We know that there exist $\beta>0$ and $\gamma\geq 0$ such that
\begin{align*}
\beta \|z_\epsilon\|_{H^1(\Omega_h^\epsilon)}\leq \int_{\Omega_h^\epsilon}(D_h\nabla z_\epsilon\cdot \nabla z_\epsilon+q^\epsilon\cdot\nabla z_\epsilon)\,dx+\gamma
\|z_\epsilon\|_{L^2(\Omega_h^\epsilon)},
\end{align*}
and we use this to estimate
\begin{align*}
&\frac{1}{2}\partial_t\|z_\epsilon\|^2_{L^2(\Omega_h^\epsilon)}+\frac{1}{2}\partial_t\|w_\epsilon\|^2_{L^2(\Omega_l^\epsilon)}
+\beta\|z_\epsilon\|^2_{H^1(\Omega_h^\epsilon)}+\epsilon^2
D_l\|\nabla w_\epsilon\|^2_{L^2(\Omega_l^\epsilon)} \leq \gamma\|z_\epsilon\|_{L^2(\Omega_h^\epsilon)}+ I_1+I_2+...+I_{11}
\end{align*}
where
\begin{align*}
I_1=&\left|\int_{\Omega^\epsilon_h}D_h \nabla_x\cdot(\nabla_xu_0+\nabla_y u_1)z_\epsilon\,dx
-\int_{\Omega^\epsilon_h}\frac{1}{\theta}\nabla_x\cdot({\cal D}(x)\nabla_x u_0)z_\epsilon\,dx\right.\\
&\left. +\epsilon D_h\int_{\Gamma^\epsilon}\nu_1\cdot(\nabla_x u_0+\nabla_yu_1)z_\epsilon\,ds\right|,\\
I_2=&\left|\int_{\Omega^\epsilon_h}(\int_{\partial B(x)}
\nu_0\cdot(D_l\nabla_yv_0)\,d\sigma)z_\epsilon\,dx-\epsilon\int_{\Gamma^\epsilon}D_l\nu_0\cdot\nabla_y v_0 z_\epsilon\,ds\right|,\\
I_3=&\left|\int_{\Omega^\epsilon_h} (1-\chi_\epsilon) D_h(\nabla_x\cdot\nabla_y u_1) z_\epsilon \,dx\right|,\\
I_4=&\left|\epsilon^2\int_{\Gamma^\epsilon}\Big(D_h(\nu_2+O(\epsilon))\cdot(\nabla_x u_0-\nabla_yu_1)+D_l\nu^\epsilon_1\cdot\nabla_yv_0\Big)z_\epsilon\,ds\right|,\\
I_5=&\left|\int_{\Omega_h^\epsilon}(\epsilon D_h \Delta\chi_\epsilon u_1+2\epsilon D_h\nabla \chi_\epsilon \cdot \nabla_x u_1
+2 D_h \nabla \chi_\epsilon \cdot \nabla_y u_1)z_\epsilon \,dx\right|,\\
I_6=&\left|\epsilon\int_{\Omega^\epsilon_h}\chi_\epsilon \partial_t u_1z_\epsilon\,dx\right|,\\
I_7=&\left|\epsilon\int_{\Omega^\epsilon_h}D_h\nabla_x u_1 \cdot \nabla(\chi_\epsilon z_\epsilon)\,dx\right|,\\
I_8=&\left|\epsilon^2\int_{\Omega^\epsilon_l}D_l \Delta_xv_0w_\epsilon\,dx
+\epsilon\int_{\Omega^\epsilon_l} D_l(\nabla_x\cdot\nabla_y v_0+\nabla_y\cdot\nabla_x v_0) w_\epsilon\,dx
-\epsilon^2\int_{\Gamma^\epsilon}D_l\nu^\epsilon\cdot\nabla_x v_0 w_\epsilon\,ds\right|,\\
I_9=&\left|\epsilon^2\int_{\Gamma^\epsilon}D_l\nu^\epsilon\cdot(\epsilon\nabla v_\epsilon-\nabla_yv_0)u_1\,ds\right|,\\
I_{10}=&\left|\int_{\Omega_h^\epsilon}(\frac{1}{\theta}\bar{q}-q_\epsilon)\cdot\nabla u_0z_\epsilon\,dx\right|,\\
I_{11}=&\left|\epsilon\int_{\Omega_h^\epsilon}\chi_\epsilon u_1 q_\epsilon \cdot \nabla z_\epsilon\,dx\right|.
\end{align*}
For $I_1$ we use that $u_1=M\nabla u_0$ and that ${\cal D}(x):=D_h\int_{Y(x)}(I+\nabla_yM)\,dy$. We set
\begin{align*}
&Q(x,y)=\nabla_x\cdot((I+\nabla_yM)\nabla_xu_0)-\frac{1}{\theta(x)}\nabla_x\cdot\int_{Y(x)}(I+\nabla_yM)\nabla_xu_0\,dy,\\
&p(x,y)=-\nu_1\cdot(I+\nabla_yM),
\end{align*}
and use Lemma \ref{lem2} to obtain $I_1\leq \epsilon C
\|z_\epsilon\|_{H^1(\Omega_h^\epsilon)}$. Lemma \ref{lem1} asserts
that the conditions of Lemma \ref{lem2} are satisfied for these
choices of $Q$ and $p$.

For $I_2$ we also apply Lemma \ref{lem2}, this time for the choice
\begin{align*}
&Q(x,y)=\frac{1}{\theta(x)}\int_{\partial B(x)}\nu_0\cdot\nabla_yv_0\,d\sigma,\\
&p(x,y)=\nu_0\cdot\nabla_yv_0,
\end{align*}
and we again get
$I_2\leq \epsilon C \|z_\epsilon\|_{H^1(\Omega_h^\epsilon)}$
with $C$ independent of $\epsilon$. For $I_3$ we have $I_3\leq \sqrt{\epsilon} C \|z_\epsilon\|_{H^1(\Omega_h^\epsilon)}$.
Application of the regularity results in Lemma \ref{addreg} and Theorem \ref{strongsol} results in
$I_{4}+I_6\leq C\epsilon\|z_\epsilon\|_{H^1(\Omega_h^\epsilon)}$.
With the use of Lemma \ref{lemch} and the properties of $\chi_\epsilon$ we
estimate $I_5+I_7\leq \sqrt{\epsilon} C \|z_\epsilon\|_{H^1(\Omega_h^\epsilon)}$. Another application of the regularity results gives $I_8 \leq \epsilon C\|w_\epsilon\|_{H^1(\Omega_l^\epsilon)}$ and $I_9\leq C\epsilon \|u_0\|_{H^2(\Omega_h^\epsilon)}(\|v_\epsilon\|_{H^2(\Omega_l^\epsilon)}+\|v_0\|_{H^2(\Omega;H^2(B(x)))})$.
For $I_{10}$ and $I_{11}$ we use the assumptions on $q$ and $q^\epsilon$ (as stated in (A2) and (B2)) to get $I_{10}+I_{11}\leq \sqrt{\epsilon}C\|z_\epsilon\|_{H^1(\Omega_h^\epsilon)}$.

Now we obtain
\begin{align*}
&\frac{1}{2}\partial_t\|z_\epsilon\|^2_{L^2(\Omega_h^\epsilon)}+\frac{1}{2}\partial_t\|w_\epsilon\|^2_{L^2(\Omega_l^\epsilon)}
+D_h\|\nabla z_\epsilon\|^2_{L^2(\Omega_h^\epsilon)}+\epsilon^2
D_l\|\nabla w_\epsilon\|^2_{L^2(\Omega_l^\epsilon)}\leq \|z_\epsilon\|_{L^2(\Omega_h^\epsilon)}+\\
&\epsilon C_1\|u_0\|_{H^2(\Omega_h^\epsilon)}(\|v_\epsilon\|_{H^2(\Omega_l^\epsilon)}+\|v_0\|_{H^2(\Omega;H^2(B(x)))})+\sqrt{\epsilon}C_2\|z_\epsilon\|_{H^1(\Omega_h^\epsilon)}+\epsilon C_3 \|w_\epsilon\|_{H^1(\Omega_l^\epsilon)}
\end{align*}
Using the energy bounds (\ref{EB1}) and  (\ref{EB2}) together with
a Gronwall-type argument lead to
\begin{align}
\|z_\epsilon\|^2_{L^\infty(I,L^2(\Omega_h^\epsilon))}+&
\|w_\epsilon\|^2_{L^\infty(I,L^2(\Omega_l^\epsilon))}+
\|z_\epsilon\|^2_{L^2(I,H^1(\Omega_h^\epsilon))}+\nonumber\\
&+\epsilon\|w_\epsilon\|^2_{L^2(I,H^1(\Omega_l^\epsilon))}\label{inest}\leq \epsilon \tilde{C}_1 +\sqrt{\epsilon}\tilde{C}_2 \|z_\epsilon\|_{L^2(I,H^1(\Omega_h^\epsilon))},
\end{align}
which, in particular, implies
\begin{align*}
\|z_\epsilon\|^2_{L^2(I,H^1(\Omega_h^\epsilon))}
\leq \epsilon \tilde{C}_1 +\sqrt{\epsilon}\tilde{C}_2 \|z_\epsilon\|_{L^2(I,H^1(\Omega_h^\epsilon))}.
\end{align*}
and thus
\begin{align*}
\|z_\epsilon\|_{L^2(I,H^1(\Omega_h^\epsilon))}
\leq \sqrt{\epsilon}\frac{1}{2}( \tilde{C}^2_2+\sqrt{\tilde{C}_2+4\tilde{C}_1}).
\end{align*}
Combining this with \eqref{inest} gives the result
\begin{align*}
\|z_\epsilon\|_{L^\infty(I,L^2(\Omega_h^\epsilon))}+
\|w_\epsilon\|_{L^\infty(I,L^2(\Omega_l^\epsilon))}+
\|z_\epsilon\|_{L^2(I,H^1(\Omega_h^\epsilon))}+&
\epsilon\|w_\epsilon\|_{L^2(I,H^1(\Omega_l^\epsilon))}\leq
c\sqrt{\epsilon},
\end{align*}
where the constant $c$ is independent of $\epsilon$.
The last step uses the evident estimate $\|\epsilon u_1(1-\chi_\epsilon)\|_{H^1(\Omega_h^\epsilon)}\leq C\sqrt{\epsilon}$, and the theorem is proven.
\end{proof}

\section*{Acknowledgments}
We acknowledge fruitful  discussions with Gregory Chechkin regarding homogenization techniques for non-periodic media.
We also thank Eduard Marusic-Paloka for an interesting correspondence on the best corrector estimates (upper bounds on convergence rates) existing
for the stationary Stokes and Navier-Stokes problems.

\bibliographystyle{siam}
\bibliography{biblio}


\end{document}